\documentclass[11pt]{article}

\usepackage{amsmath,amsfonts,amsthm,fullpage,authblk}
\newcommand\E{\mathbb{E}}
\newcommand\R{\mathbb{R}}
\renewcommand\t{{\scriptscriptstyle\top}}

\newcommand\eps{\epsilon}
\newcommand\veps{\varepsilon}
\newcommand\tl{\tilde}
\newcommand\wh{\widehat}

\newcommand\tr{\operatorname{tr}}
\newcommand\rank{\operatorname{rank}}
\newcommand\diag{\operatorname{diag}}

\newtheorem{lemma}{Lemma}
\newtheorem{theorem}{Theorem}

\title{Analysis of a randomized approximation scheme for matrix multiplication}
\author[1]{Daniel Hsu}
\author[1]{Sham M. Kakade}
\author[2]{Tong Zhang}
\affil[1]{Microsoft Research, New England}
\affil[2]{Department of Statistics, Rutgers University}

\begin{document}
\maketitle
{\def\thefootnote{}
\footnotetext{E-mail:
\texttt{dahsu@microsoft.com},
\texttt{skakade@microsoft.com},
\texttt{tzhang@stat.rutgers.edu}}

\begin{abstract}
This note gives a simple analysis of a randomized approximation scheme for
matrix multiplication proposed by~\cite{Sarlos06} based on a random
rotation followed by uniform column sampling.
The result follows from a matrix version of Bernstein's inequality and a
tail inequality for quadratic forms in subgaussian random vectors.
\end{abstract}

\section{Introduction}

Let $A := [a_1|a_2|\dotsb|a_m] \in \R^{d_A \times m}$ and $B :=
[b_1|b_2|\dotsb|b_m] \in \R^{d_B \times m}$ be fixed matrices, each with $m$
columns.
If $m$ is very large, then the straightforward computation of the matrix
product $AB^\t$ (with $\Omega(d_A d_B m)$ operations) can be prohibitive.

We can instead approximate the product using the following randomized
scheme.
Let $\Theta \in \R^{m \times m}$ be a random orthogonal matrix; the
distribution of $\Theta$ will be specified later in Theorem~\ref{thm:main},
but a key property of $\Theta$ will be that the matrix products
\[ \tl{A} := A\Theta \qquad\text{and}\qquad \tl{B} := B\Theta \]
can be computed with $O((d_A + d_B) m \log m)$ operations.
Given the products $\tl{A} = [\tl{a}_1|\tl{a}_2|\dotsb|\tl{a}_m]$
and $\tl{B} = [\tl{b}_1|\tl{b}_2|\dotsb|\tl{b}_m]$, we take a small uniform
random sample of pairs of their columns (drawn with replacement)
\[ (\tl{a}_{i_1}, \tl{b}_{i_1}), (\tl{a}_{i_2}, \tl{b}_{i_2}), \dotsc,
(\tl{a}_{i_n}, \tl{b}_{i_n}) , \]
and then compute the sum of outer products
\[ \wh{AB^\t} := \frac{m}{n} \sum_{j=1}^n \tl{a}_{i_j} \tl{b}_{i_j}^\t . \]
It is easy to check that $\wh{AB^\t}$ is an unbiased estimator of $AB^\t$.
The sum can be computed from $\tl{A}$ and $\tl{B}$ with $O(d_A d_B n)$
operations, so overall, the matrix $\wh{AB^\t}$ can be computed with $O(d_A
d_B n + (d_A + d_B) m \log m)$ operations.
(In fact, the $\log m$ can be replaced by $\log n$~\cite{AilLib08}.)
Therefore, we would like $n$ to be as small as possible so that, with high
probability, $\|\wh{AB^\t} - AB^\t\| \leq \veps \|A\| \|B\|$ for some error
$\veps > 0$, where $\|\cdot\|$ denotes the spectral norm.
As shown in Theorem~\ref{thm:main}, it suffices to have
\[ n = \Omega\biggl( \frac{(k + \log(m)) \log(k)}{\veps^2} \biggr) , \]
where $k := \max\{ \tr(A^\t A) / \|A\|^2, \ \tr(B^\t B) / \|B\|^2 \} \leq
\max\{\rank(A), \rank(B)\}$.

A flawed analysis of a different scheme based on non-uniform column
sampling (without a random rotation $\Theta$) was given
in~\cite{HKZ12-matrix}; that analysis gave an incorrect bound on
$\|\E[X^2]\|$ for a certain random symmetric matrix $X$.
A different analysis of this non-uniform sampling scheme can be found
in~\cite{MagZou11}, but that analysis has some deficiencies as pointed out
in~\cite{HKZ12-matrix}.
The scheme studied in the present work, which employs a certain random
rotation followed by uniform column sampling, was proposed by
\cite{Sarlos06}, and is based on the Fast Johnson-Lindenstrauss Transform
of~\cite{AilCha09}.
The analysis given in \cite{Sarlos06} bounds the Frobenius norm error; in
this work, we bound the spectral norm error.
A similar but slightly looser analysis of spectral norm error was very
recently provided in \cite{ABTZ12}.

\section{Analysis}

Let $[m] := \{ 1, 2, \dotsc, m \}$.

\begin{theorem} \label{thm:main}
Pick any $\delta \in (0,1/3)$, and let $k := \max\{ \tr(AA^\t) / \|A\|^2, \
\tr(BB^\t) / \|B\|^2 \}$ (note that $k \leq \max\{ \rank(A), \ \rank(B)
\}$.
Assume $\Theta = \frac1{\sqrt{m}} DH$, where $D = \diag(\eps)$, $\eps \in
\{\pm1\}^m$ is a vector of independent Rademacher random variables, and $H
\in \{\pm1\}^{m \times m}$ is a Hadamard matrix.
With probability at least $1 - \delta$,
\begin{multline*}
\|\wh{AB^\t} - AB^\t\| \leq
\|A\| \|B\| \Biggl(
\sqrt{\frac{4(k + 2\sqrt{k\ln(3m/\delta)} + 2\ln(3m/\delta) +
1)\ln(6k/\delta)}{n}} \\
+ \frac{2(k + 2\sqrt{k\ln(3m/\delta)} + 2\ln(3m/\delta)
+ 1)\ln(6k/\delta)}{3n}
\Biggr)
.
\end{multline*}
\end{theorem}

The proof of Theorem~\ref{thm:main} is a consequence of the following
lemmas, combined with a union bound.
Lemma~\ref{lem:main} bounds the error in terms of a certain quantity $\mu$
which depends on the random orthogonal matrix $\Theta$ (and $A$ and $B$).
Lemma~\ref{lem:coherence} gives a bound on $\mu$ that holds with high
probability over the random choice of $\Theta$.

\begin{lemma} \label{lem:main}
Define
$Q = [q_1|q_2|\dotsb|q_m] := \|A\|^{-1} A \Theta$,
$R = [r_1|r_2|\dotsb|r_m] := \|B\|^{-1} B \Theta$,
$k_A := \tr(QQ^\t) = \tr(AA^\t) / \|A\|^2$,
$k_B := \tr(RR^\t) = \tr(BB^\t) / \|B\|^2$, and
\begin{align*}
\mu & := m \max \Bigl(
\bigl\{ \|q_i\|^2 : i \in [m] \bigr\}
\ \cup \ \bigr\{ \|r_i\|^2 : i \in [m] \bigr\} \Bigr)
.
\end{align*}
Then
\[
\Pr\biggl[ \|\wh{AB^\t} - AB^\t\| >
\|A\| \|B\| \biggl(
\sqrt{\frac{2(\mu + 1)t}{n}} + \frac{(\mu + 1)t}{3n}
\biggr) \biggr] \leq 2\sqrt{k_A k_B}
\cdot \frac{t}{e^t - t - 1}
.
\]
\end{lemma}
\begin{proof}
Observe that because $\Theta$ is orthogonal,
\begin{align*}
\|\wh{AB^\t} - AB^\t\|
= \|A\| \|B\| \biggl\| \frac{m}{n} \sum_{j=1}^n q_{i_j} r_{i_j}^\t -
Q R^\t \biggr\|
.
\end{align*}
We now derive a high probability bound for the last term on the right-hand
side.
Define a random symmetric matrix $X$ with
\[ \Pr\biggl[ X = m \begin{bmatrix} 0 & q_i r_i^\t \\ r_i q_i^\t
& 0 \end{bmatrix} \biggr] = \frac1m , \quad i \in [m] , \]
and let $X_1, X_2, \dotsc, X_n$ be independent copies of $X$.
Define
\[ \wh{M} := \frac1n \sum_{j=1}^n X_j
\qquad \text{and} \qquad
M := \begin{bmatrix} 0 & Q R^\t \\ RQ^\t & 0
\end{bmatrix} . \]
Then
\[ \|\wh{M} - M\|
= \biggl\| \frac1n \sum_{j=1}^n (X_j - M) \biggr\|
\overset{\text{distribution}}{=}
\biggl\| \frac{m}{n} \sum_{j=1}^n q_{i_j} r_{i_j}^\t - Q R^\t
\biggr\| . \]
Observe that $\E[ X - M ] = 0$ and $\|X - M\| \leq \|X\| + \|M\| \leq \mu +
1$.
Moreover,
\begin{align*}
\E[X]^2
& = M^2
= \begin{bmatrix} Q R^\t R Q^\t & 0 \\ 0 & R Q^\t Q R^\t
\end{bmatrix}
, \\
\E[X^2]
& = \sum_{i=1}^m m \begin{bmatrix} \|r_i\|^2 q_iq_i^\t & 0 \\ 0 &
\|q_i\|^2 r_ir_i^\t \end{bmatrix}
= m \begin{bmatrix} \sum_{i=1}^m \|r_i\|^2 q_iq_i^\t & 0 \\ 0 &
\sum_{i=1}^m \|q_i\|^2 r_ir_i^\t \end{bmatrix}
,
\\
\tr(\E[X^2])
& = 2m \sum_{i=1}^m \|q_i\|^2 \|r_i\|^2
\leq 2\mu \sum_{i=1}^m \|q_i\| \|r_i\|
\leq 2\mu \sqrt{k_A k_B}
,
\\
\|\E[X^2]\|
&
\leq m \max\Bigl\{ \Bigl\| \sum_{i=1}^m \|r_i\|^2 q_iq_i^\t \Bigr\|, \
\Bigl\| \sum_{i=1}^m \|q_i\|^2 r_ir_i^\t \Bigr\| \Bigr\}
\leq \mu \max\Bigl\{ \| QQ^\t \|, \ \| RR^\t \| \Bigr\}
= \mu
,
\\
\|\E[(X - M)^2]\|
& = \|\E[X^2] - M^2\| \leq \mu + 1
.
\end{align*}
Therefore, by the matrix Bernstein inequality from~\cite{HKZ12-matrix},
\[
\Pr\biggl[ \|\wh{M} - M\| >
\sqrt{\frac{2(\mu + 1)t}{n}} + \frac{(\mu + 1)t}{3n}
\biggr] \leq 2\sqrt{k_A k_B} \cdot \frac{t}{e^t - t - 1} .
\]
The lemma follows.
\end{proof}

The following lemma is a special case of a result found
in~\cite{HKZ11-regression}.
\begin{lemma} \label{lem:coherence}
Assume $\Theta = \frac1{\sqrt{m}} DH$, where $D = \diag(\eps)$, $\eps \in
\{\pm1\}^m$ is a vector of independent Rademacher random variables, and $H
\in \{\pm1\}^{m \times m}$ is a Hadamard matrix.
Let $Z \in \R^{m \times d}$ be a matrix with $\|Z\| \leq 1$, and set $k_Z
:= \tr(ZZ^\t)$.
Then
\[
\Pr\biggl[ \max\{ \|Z^\t \Theta e_i\|^2 : i \in [m] \}
> \frac{k_Z + 2\sqrt{k_Z(\ln(m) + t)} + 2(\ln(m) + t)}{m} \biggr]
\leq e^{-t}
\]
where $e_i \in \{0,1\}^m$ is the $i$-th coordinate axis vector in $\R^m$.
\end{lemma}
\begin{proof}
Observe that for each $i \in [m]$, the random vector $\sqrt{m}
\Theta e_i$ has the same distribution as $\eps$.
Moreover, $\eps$ is a subgaussian random vector in the sense that $\E[
\exp(\alpha^\t \eps) ] \leq \exp(\|\alpha\|^2 / 2)$ for any vector $\alpha
\in \R^m$.
Therefore, we may apply a tail bound for quadratic forms in subgaussian
random vectors (\emph{e.g.},~\cite{HKZ12-quadratic}) to obtain
\[
\Pr\biggl[ \|\sqrt{m} Z^\t \Theta e_i\|^2
> \tr(ZZ^\t) + 2\sqrt{\tr((ZZ^\t)^2)\tau} + 2\|ZZ^\t\|\tau
\biggr] \leq e^{-\tau}
\]
for each $i \in [m]$ and any $\tau > 0$.
The lemma follows by observing that $\|ZZ^\t\| \leq 1$ and $\tr((ZZ^\t)^2)
\leq \tr(ZZ^\t) \|ZZ^\t\| \leq k_Z$, and applying a union bound over all $i
\in [m]$.
\end{proof}

We note that Lemma~\ref{lem:coherence} holds for many other distributions
of orthogonal matrices (with possibly worse constants).
All that is required is that $\sqrt{m} \Theta e_i$ be a subgaussian random
vector for each $i \in [m]$.
See~\cite{HKZ11-regression} for more discussion.

\begin{proof}[Proof of Theorem~\ref{thm:main}]
We apply Lemma~\ref{lem:coherence} with both $Z = A / \|A\|$ and $Z = B /
\|B\|$, and combine the implied probability bounds with a union bound to
obtain
\[ \Pr\bigl[ \mu > k + 2\sqrt{k\log(3m/\delta)} + 2\ln(3m/\delta) \bigr]
\leq 2\delta/3 , \]
where $\mu$ is defined in the statement of Lemma~\ref{lem:main}, and the
probabiltiy is taken with respect to the random choice of $\Theta$.
Now we apply Lemma~\ref{lem:main}, together with the bound $t/(e^t-t-1)
\leq e^{-t/2}$ for $t \geq 2.6$, and substitute $t := 2\ln(6k/\delta)$ to
obtain
\[
\Pr\biggl[ \|\wh{AB^\t} - AB^\t\| >
\|A\| \|B\| \biggl(
\sqrt{\frac{4(\mu + 1)\ln(6k/\delta)}{n}}
+ \frac{2(\mu + 1)\ln(6k/\delta)}{3n}
\biggr) \biggr] \leq \delta/3
.
\]
Combining the two probability bounds with a union bound implies the claim.
\end{proof}

\subsection*{Acknowledgements}

We thank Joel Tropp for pointing out the error in the analysis
from~\cite{HKZ12-matrix}.

\subsection*{References}
{\def\section*#1{} \bibliography{mult} \bibliographystyle{alpha}}

\end{document}